\documentclass[pra,twocolumn,showpacs]{revtex4-2}

\usepackage{amsmath}
\usepackage{latexsym}
\usepackage{amssymb}
\usepackage{graphicx}
\usepackage[colorlinks=true, citecolor=blue, urlcolor=blue]{hyperref}
\usepackage{float}
\usepackage{amsfonts}
\usepackage{textcomp}
\usepackage{mathpazo}
\usepackage{comment}
\usepackage{xr}

\usepackage{bbm}

\usepackage{xcolor}
\definecolor{myurlcolor}{rgb}{0,.5,.5}
\definecolor{mycitecolor}{rgb}{0,.6,0}
\definecolor{myrefcolor}{rgb}{2,0,0}
\usepackage{hyperref}
\hypersetup{colorlinks,
linkcolor=myrefcolor,
citecolor=mycitecolor,
urlcolor=myurlcolor}

\usepackage[draft]{fixme}
\usepackage{amsmath,bbm}
\usepackage{graphicx}
\usepackage{amsfonts}
\usepackage{amssymb}
\usepackage{amsmath, amssymb, amsthm,verbatim,graphicx,bbm}
\usepackage{mathrsfs}
\usepackage{color,xcolor,longtable}

\usepackage{xr}
\makeatletter
\newcommand*{\addFileDependency}[1]{
  \typeout{(#1)}
  \@addtofilelist{#1}
  \IfFileExists{#1}{}{\typeout{No file #1.}}
}
\makeatother

\newcommand*{\myexternaldocument}[1]{
    \externaldocument{#1}
    \addFileDependency{#1.tex}
    \addFileDependency{#1.aux}
}

\myexternaldocument{manuscript_PRL}

\newcommand{\beq}[0]{\begin{equation}}
\newcommand{\eeq}[0]{\end{equation}}

\newcommand{\one}{\leavevmode\hbox{\small1\normalsize\kern-.33em1}}

\def\be{\begin{equation}}
\def\ee{\end{equation}}
\def\ben{\begin{eqnarray}}
\def\een{\end{eqnarray}}
\def\eea{\end{array}}
\def\bea{\begin{array}}

\newcommand{\Tr}[1]{\mathrm{Tr}#1}
\newcommand{\bei}{\begin{itemize}}
\newcommand{\eei}{\end{itemize}}
\newcommand{\ket}[1]{|#1\rangle}
\newcommand{\bra}[1]{\langle#1|}

\newcommand{\proj}[1]{\ket{#1}\!\!\bra{#1}}

\newcommand{\I}{\mathbbm{1}}

\newcommand{\p}{\Vec{p}}

\renewcommand{\emph}[1]{\textbf{#1}}

\makeatletter
\newtheorem*{rep@theorem}{\rep@title}
\newcommand{\newreptheorem}[2]{%
\newenvironment{rep#1}[1]{%
 \def\rep@title{#2 \ref{##1}}%
 \begin{rep@theorem}}%
 {\end{rep@theorem}}}
\makeatother

\theoremstyle{plain}
\newtheorem{thm}{Theorem}
\newtheorem*{thm*}{Theorem}
\newreptheorem{thm}{Theorem}

\newtheorem{fakt}{Fact}

\newtheorem{defn}[thm]{Definition}

\theoremstyle{definition}

\theoremstyle{remark}

\usepackage[T1]{fontenc}

\begin{document}

\title{Network steering with arbitrarily low detection efficiency of any entangled measurement}
\author{Shubhayan Sarkar}
\email{shubhayan.sarkar@ug.edu.pl}
\affiliation{Institute of Informatics, Faculty of Mathematics, Physics and Informatics,
University of Gdansk, Wita Stwosza 57, 80-308 Gdansk, Poland}

\begin{abstract}	
Quantum nonlocality and its asymmetric counterpart, quantum steering, are among the most intriguing manifestations of quantum mechanics. From a theoretical perspective, they are not only of fundamental significance but also hold promise for a wide range of applications. However, the stringent technological requirements for their experimental observation in a loophole-free way have largely restricted their realization to foundational demonstrations. A major challenge in these setups is the limited detection efficiency of current detectors as to observe nonlocality or steering in the standard scenarios, one requires detectors above a certain critical efficiency which can only be achieved with superconducting detectors. Considering the simplest quantum network, we demonstrate here that quantum steering between two parties, can be demonstrated for any non-zero detection efficiency, if the sources in the network generate states above a critical visibility well-within the current practical limits. This form of quantum steering in networks, is termed swap-steering. Moreover, when the sources are perfect, swap-steering can be observed using any entangled measurement on the untrusted side, even with arbitrarily low detection efficiency. Consequently, two major loopholes, the detection-loophole as well as free-will loophole can be closed easily in quantum steering experiments, when implemented using networks.
\end{abstract}


\maketitle

{\it{Introduction---}} Quantum nonlocality \cite{Bell,Bell66} stands as one of the most paradigm-shifting concepts to have emerged in modern physics that has not just reshaped our foundational understanding of the physical world but also has led to several useful applications such as device-independent (DI) quantum information protocols, including DI cryptography \cite{DICrypto,NonlocalityReview}. While the theoretical side of nonlocality along with the applications have been extensively explored, observing them in a loophole-free way still requires state-of-art quantum infrastructure such as superconducting detectors, genuine random number generators along with sources that generate expected states with high visibility \cite{exp1,exp2,exp3,exp4,exp5,exp6,exp7}. As a matter of fact, beyond the simplest scenarios, quantum nonlocality or quantum steering (assymetric form of quantum nonlocality) \cite{Wiseman} has not yet been demonstrated in a loophole-free way. This limits the applicability and explorability of quantum nonlocality from a practical perspective.

There are three major loopholes in any nonlocality (or steering) experiments \cite{NonlocalityReview}. First, locality-loophole, which states that the measurements on both subsystems should be fast enough such that no signal can reach from one to other. This loophole can be closed by large spatial distance and fast electronics well-within the current practical reach. Second, the detection-loophole, which states that below a certain critical efficiency of the detectors, there exist local models to explain the observed correlations. To close this loophole, one requires superconducting detectors with high efficiency. Third, the free-will loophole, where measurement settings on both subsystems need to chosen independent each other, while having the capability of switching them at will during the course of the experiment. The latter two loopholes pose significant experimental challenges for demonstrating nonlocality (or steering), as closing them requires stringent technological requirements, making loophole-free implementations extremely demanding in practice. Consequently, exploring scenarios that enable the observation of nonlocality or steering while imposing less stringent technological requirements and offering greater resilience to noise is an important direction for further research.

Quantum networks with multiple sources provide a novel perspective on this problem. Here, we delve into a recently introduced notion of network steering \cite{netstee}, called swap-steering \cite{Sarkar2024networkquantum}, where two parties can observe quantum steering, one of them being trusted, with single fixed measurements, thus readily closing the free-will loophole. We find that allowing for tomographically complete measurements on the trusted side allows the observation of quantum steering in this network with arbitrary low detection efficiency of measurements on the untrusted side, if the sources generate states above a certain critical visibility. We then show that if the sources are perfect, then by using any entangled measurement on the untrusted side with arbitrary low detection efficiency one can observe swap-steering. Considering parameters from current practical experiments, we finally conclude that, quantum steering in networks can be demonstrated using room-temperature detectors in a loophole-free way that does not even require genuine random number generators.

For our purpose, we first adapt the idea of process tomography \cite{tomo1, tomo2, tomo3, tomo4,tomo5,tomo6} to witness entanglement of any measurement. Then, using the corresponding witness, we construct the swap-steering witnesses, allowing us to observe steering in quantum networks among two parties, using any entangled measurement. Consequently, we provide a one-sided DI witness of entanglement of any measurement. Finally, considering realistic detectors and sources, we find the critical visibility beyond which swap-steering can be demonstrated using any entangled measurement.

{\it{Quantum measurement tomography---}} We begin by describing the tomography of quantum measurements. 
Consider a quantum measurement $\{M_i\}$ where $M_i$ denotes the measurement elements such that they are positive semi-definite and $\sum_iM_i=\I$. For projective measurements, an additional condition is $M_iM_j=M_i\delta_{ij}$. Considering a set of informationally complete density matrices $\rho_j$ that are known (trusted), one can reconstruct the measurement by evaluating the probabilities, also referred to as correlations, $p_{ij}=\Tr(M_i\rho_j)$. 
Here one additionally assumes the dimension of the Hilbert space on which the measurements act. A broader class of tomography is known as process tomography where one considers tomography of quantum channels and has to trust the input states along with measurements \cite{tomo1, tomo2, tomo3, tomo4,tomo5,tomo6}.

At most times, it will not be required to reconstruct the complete measurement but only to confirm some relevant properties about it. 
Here we are interested in the entanglement in composite quantum measurements and thus classify them into two classes:

\begin{defn}[Separable measurements] \label{defsep} Consider the measurement $M_i$ such that all the measurement elements $\{M_i\}$ are separable, that is, $M_i\propto \sum_j\sigma_{i,j}\otimes\sigma'_{i,j}\otimes\ldots\ \forall i$, then  $\{M_i\}$ is a separable measurement. 
\end{defn}
Using the above definition, we define entangled measurements.
\begin{defn}[Entangled measurements] \label{defent} Consider the measurement $\{E_i\}$ such that at least one of the  measurement elements $E_i$ is not separable, then  $\{E_i\}$ is an entangled measurement. 
\end{defn}

Let us now construct witnesses that can detect whether a given composite measurement is separable or entangled.  

{\it{Witnesses---}} It is well-known that due to Hahn-Banach theorem, for every entangled state $\rho_e$ there exists an entanglement witness $\mathcal{W}_{\rho_e}$ such that $\Tr(\mathcal{W}_{\rho_e}\rho_e)<0$ and $\Tr(\mathcal{W}_{\rho_e}\sigma)\geq0$ where $\sigma$ represents any separable state. Similarly, we can find a general result to witness entanglement in quantum measurements that straightaway follows from the Hahn-Banach theorem.
\begin{fakt} For any entangled measurement $\{E_i\}$ def. \ref{defent}, there exists an entanglement witness $\mathcal{W}$ such that 
\begin{equation}\label{eq1}
   \min_i\Tr(\mathcal{W}E_i)<0,\quad \min_i\Tr(\mathcal{W}M_i)\geq0
\end{equation}
for every separable measurement $\{M_i\}$.
\end{fakt}
\begin{proof}
    Consider that the element $E_k$ of $\{E_i\}$ is entangled. Considering the witness $\mathcal{W}_{E_k}$, we have that $\Tr(\mathcal{W}_{E_k}E_k)<0$ and $\Tr(\mathcal{W}_{E_k}\sigma)\geq0$ for every separable state $\sigma$. As witnesses exist for every entangled state, thus every entangled measurement can be witnessed using the construction \eqref{eq1}.
\end{proof}

\begin{figure}[t]
\begin{center}
\includegraphics[width=.6\linewidth]{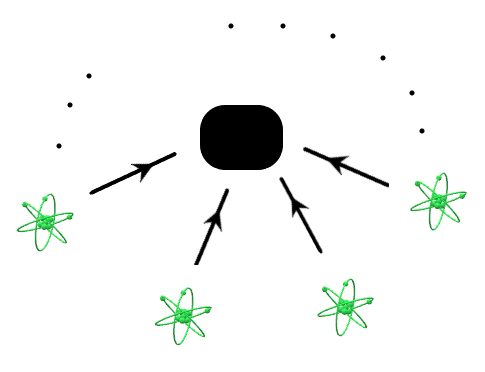}
    \caption{Witnessing entanglement in measurements. To detect entanglement in the measurement acting on $N$ subsystems, one needs to consider $N$ independent trusted sources that generate local states and send them to the untrusted measurement. From the obtained statistics, one can infer whether the measurement has entanglement or not.}
    \label{fig1}
    \end{center}
\end{figure}
We do not delve into the details of further construction of the witnesses of entanglement in measurements as it follows directly from constructing witnesses of entangled states which has been extensively studied \cite{QUANTUMENT, Guhne_2009}. We focus here on the experimental implementation of these witnesses. As an example, let us consider the Bell-basis measurement (BM) $\{\proj{\phi_{+}},\proj{\phi_{-}},\proj{\psi_{+}},\proj{\psi_{-}}\}$ where 
\begin{equation}\label{Amea1}
    \ket{\phi_{\pm}}=\frac{1}{\sqrt{2}}\left(\ket{00}\pm\ket{11}\right),\quad
    \ket{\psi_{\pm}}=\frac{1}{\sqrt{2}}\left(\ket{01}\pm\ket{10}\right).
\end{equation}
Now, to detect entanglement in Bell-basis measurement, a possible witness $\mathcal{W}_{BM}$ is given by $\mathcal{W}_{BM}=1/2\I-\proj{\phi_+}$. A way to implement this witness is to prepare the state $\ket{\phi_+}$ and act it on the measurement and if one satisfies the criterion \eqref{eq1}, then the measurement is confirmed to be entangled. However, entangled states is a costly resource and preparing particular entangled states is even more difficult. Thus, it is beneficial to express this witness in terms of local states which is much simpler to generate [see Fig. \ref{fig1}]. 

Any witness $\mathcal{W}$ can be decomposed into the form 
\begin{equation}
    \mathcal{W}=\sum_{ij\ldots} c_{ij\ldots}\proj{a_i}\otimes\proj{b_j}\otimes\ldots
\end{equation}
where $\proj{a_i},\proj{b_j},\ldots$ are positive semi-definite for any $i,j,\ldots$. Normalising and absorbing the factors into $c_{ij\ldots}$, we can safely conclude that they are valid density matrices. However, this might not be optimal in the sense that one might be required to measure a large number of correlations to conclude that the measurement is entangled. For instance, the witness $\mathcal{W}_{BM}$ decomposes as
\begin{equation}
    \mathcal{W}_{BM}=\frac{1}{4}\left(\I-\sum_{i=0}^2\sum_{j,j'=0,1}(-1)^{j+j'}\proj{\mu_{i,j}}\otimes\proj{\mu_{i,j'}}\right)
\end{equation}
where $\ket{\mu_{0,j}}=\ket{j},\ket{\mu_{1,j}}=1/\sqrt{2}(\ket{0}+(-1)^j\ket{1}),$ and $\ket{\mu_{2,j}}=1/\sqrt{2}(\ket{0}+i(-1)^j\ket{1})$. Thus, at least $12$ correlations must be considered to evaluate whether the BM is entangled. 

As it turns out, the number of correlations to observe in the particular case of BM can be significantly reduced. Consider the following witness 
\begin{eqnarray}
     \mathcal{W'}_{BM}=\frac{3}{2}\I-\sum_{i,j=0,1}\proj{\mu_{i,j}}\otimes\proj{\mu_{i,j}}.
\end{eqnarray}
It is simple to check that $\Tr(\mathcal{W'}_{BM}\sigma)\geq0$ for any separable state $\sigma$ and $\Tr(\mathcal{W'}_{BM}\proj{\phi_+})=-1/2$. Moreover one has to measure only $4$ correlations now. We do not prove here that the witness $\mathcal{W'}_{BM}$ is optimal. However, based on our extensive search this seems to be the case. Thus, we conjecture that the above witness $\mathcal{W'}_{BM}$ is an optimal witness of the BM. Similarly, one can find different optimal witnesses tailored to specific quantum measurements that can be implemented using local quantum states.

{\it{One-sided device-independent witness---}} Let us now proceed towards relaxing the trust in the input states to show that a given measurement is entangled. For this purpose, we utilise the recently contrived swap-steering scenario, which is the minimal scenario to detect any form of network nonlocality without inputs \cite{Sarkar2024networkquantum}. The swap-steering scenario consists of two parties namely, Alice and Bob in two different labs. Both of them receive subsystems from $N$ classically correlated sources $S_j$ that prepare the states $\bigotimes_{j=1}^N\rho_{A_jB_j}^{(k)}$ for $j=1,\ldots,N$ and mix them with a probability $p_k$. Here $A_j, B_j$ denote the $N$ different subsystems of Alice and Bob respectively. Bob performs a single measurement $\{E_b\}$ on the received subsystems where $b$ denotes his outcomes such that $E_b$ acts on $\bigotimes_N\mathbb{C}^d$. Alice is trusted in this context, meaning the measurements she performs on her subsystems are well-defined and known. Here, we consider the measurements of the trusted party is $\mathcal{A}_{i_1\ldots i_N}=\{\tau_{i_1}\otimes\ldots\tau_{i_N},\I-\tau_{i_1}\otimes\ldots\tau_{i_N}\}$ where $\tau_{i_j}\in \mathbb{C}^{d}$ are projectors such that $\{\tau_{i_j}\}$ span the Hilbert spaces $\mathbb{C}^{d}$ and thus are tomographically complete. Consequently, we have that $i_j=\{1,\ldots,d^2\}$ for $j=1,\ldots,N$. As we trust one of the sides in the experiment, in general such scenarios are referred to as one-sided device-independent (1SDI) (see Fig. \ref{fig2}). 

Alice and Bob repeat the experiment enough times to obtain the correlations $\vec{p}=\{p(a,b|i_1\ldots i_N)\}$ where $p(a,b|i_1\ldots i_N)$ denotes the probability of obtaining outcome $a,b$ with Alice and Bob respectively given Alice's input $i_1\ldots i_N$. These probabilities can be computed in quantum theory as
\begin{equation}
p(a,b|i_1\ldots i_N)=\sum_{k}p_k\Tr\left[(M_{a|i_1\ldots i_N}\otimes E_b)\bigotimes_{i=1}^N\rho_{A_iB_i}^{(k)}\right]
\end{equation}
where $M_{a|x}$ denote the measurement elements of Alice corresponding to input $x$. 
It is important to recall that Alice and Bob can not communicate with each other during the experiment. 
\begin{figure}[t]
\begin{center}
\includegraphics[width=.6\linewidth]{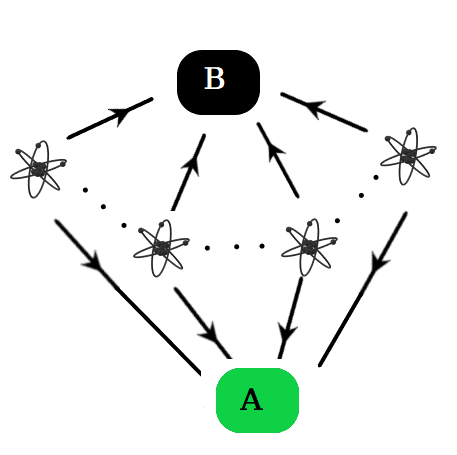}
    \caption{Swap-steering scenario. Alice and Bob are spatially separated and each of them receives $N$ subsystems from the untrusted sources. On the received subsystem Bob performs a single measurement while Alice being trusted performs a tomography on her subsystem. From the experiment, they obtain the correlations $\{p(a,b|i_1\ldots i_N)\}$.}
    \label{fig2}
    \end{center}
\end{figure}
If the correlations $\vec{p}$ admit a separable outcome-independent hidden state (SOHS) model \cite{Sarkar2024networkquantum, sarkar2024witnessingnetworksteerabilitybipartite}, then $p(a,b|i_1\ldots i_N)$ is given by
\begin{eqnarray}\label{sohs}
\sum_{k}p_k\sum_{\lambda_1,\ldots,\lambda_N}p(\lambda_1^{(k)},\ldots,\lambda_N^{(k)})\Tr{[M_{a|i_1\ldots i_N}\rho_{\lambda_1^{(k)}}\otimes\ldots\rho_{\lambda_N}^{(k)}]} \nonumber\\p(b|\lambda_1^{(k)},\ldots,\lambda_N^{(k)}).\qquad
\end{eqnarray}
Unlike the standard network considerations, here the sources are not assumed to be independent but can be classically correlated [see \cite{Sarkar2024networkquantum}]. 

Consider now that the sources generate the classically correlated state given by $\sum_{k}p_k\bigotimes_{i=1}^N\rho^{(k)}_{A_iB_i}$ and Bob's measurement $\{M_b\}$ is separable and given by $M_b=\sum_{j}\bigotimes_{i=1}^N\sigma^{(j)}_{i,b}$ such that 
$b=0,\ldots,n-1$. We show in Appendix A that for every such quantum realisation, one can construct an SOHS model of the form \eqref{sohs} that reproduces the same statistics deterministically. 
Consequently, every separable measurement generates correlations in the swap-steering network [Fig. \ref{fig2}] that can be explained via an SOHS model.

Consider now an entangled measurement $\{E_b\}$ whose entanglement witness is given by $\mathcal{W}_{\{E_b\}}$ such that it satisfies the criterion \eqref{eq1}. The operator $\mathcal{W}_{\{E_b\}}$ can be expressed using the tomographically complete set of operators stated above as $\mathcal{W}_{\{E_b\}}=-\sum_{i_1,\ldots i_N=1}^{d^2}\beta_{i_1,\ldots i_N}\tau_{i_1}\otimes\ldots\tau_{i_N}$. Utilising this entanglement witness, let us now propose the following swap-steering inequalities 
\begin{eqnarray}\label{Witunivm}
    \mathcal{S}_{\{E_b\}}=\max_{b}\sum_{i_1,\ldots i_N}\beta_{i_1,\ldots i_N}p(0,b|i_1\ldots i_N)
\end{eqnarray}
where $p(0,b|i_1\ldots i_N)$ is the probability of obtaining outcome $0$ by trusted Alice given the input $i_1\ldots i_N$ and outcome $b$ by Bob.
Let us now find its SOHS bound.
\begin{fakt}
    Consider the swap-steering scenario in Fig. \ref{fig2} and the functional $\mathcal{S}_{\{E_b\}}$ \eqref{Witunivm}. The maximal value attainable of $\mathcal{S}_{\{E_b\}}$ using an SOHS model is $\beta_{SOHS}=0$. 
\end{fakt}

The proof of the above fact is in the Appendix B.
Let us now consider that all the sources produce the maximally entangled state of dimension $d$, that is,  $\ket{\psi_{A_i,B_i}}=\ket{\phi^+_d}=1/\sqrt{d}(\sum_{i}\ket{ii})$ with Bob performing the measurement $\{E_b\}$. Let us say that Bob obtains the outcome $b$. Due to entanglement swapping, the post-measurement state at Alice is $E_b$ with a probability $\Tr(E_b)/d^{N}$ and thus one is left with the expression
\begin{eqnarray}\label{eq9}
    \mathcal{S}_{\{E_b\}}&=&\max_{b}\frac{\Tr(E_b)}{d^N}\sum_{i_1,\ldots i_N}\beta_{i_1,\ldots i_N}\Tr(\tau_{i_1}\otimes\ldots\tau_{i_N} E_b)\nonumber\\&=&\max_{b}\frac{\Tr(E_b)}{d^N}\Tr(-\mathcal{W}E_b)>0.
\end{eqnarray} 
Consequently, we can conclude that every entangled measurement generates swap-steerable correlations and thus entanglement in any quantum measurement can be witnessed in a 1SDI way. For simplicity, we considered that the dimension of the local subsystems on which $\{E_b\}$ acts is the same. However, all the above arguments can straightaway be generalised to the case when they are different by considering that the source generates the maximally entangled state of the corresponding dimension of the local subsystem and Alice performs tomography on it.

{\it{Realistic scenarios---}} 
Detector inefficiency is one of the principal obstacles to the experimental certification of quantum nonlocality. 
If the overall efficiency falls below a critical threshold, the observed statistics can often be explained by local hidden-variable models, giving rise to the well-known detection loophole \cite{Pearle1970, Garg1987}. Consequently, standard Bell and steering tests typically require high detection efficiencies to certify nonclassical correlations \cite{CH, Branciard2012, Bennet2012}. This challenge is particularly relevant for device-independent and semi-device-independent quantum information protocols, whose security fundamentally relies on loophole-free demonstrations of nonlocality \cite{Acin2007, Brunner2014}. Developing methods that remain robust against arbitrarily large losses is therefore of both foundational and practical importance.

Let us suppose now that the detectors can only detect all subsystems together with a probability $\eta$. We then define an additional outcome $\emptyset$ that takes into account all the non-simultaneous-detection events. Consequently, the measurement of Bob in the noisy scenario is given by $ \{E_b^{\mathrm{noisy}},E_{\emptyset}\}$ with $  E_b^{\mathrm{noisy}}=\eta E_b$ which further allows us to conclude that $E_{\emptyset}=(1-\eta)\I$ \cite{NonlocalityReview}. Notice that as Alice is trusted, she can do post-processing on her end, that is, she can faithfully ignore the non-simultaneous-detection events \cite{Srivastav_2022, Masini}. Moreover, let us also consider realistic sources such that the expected state is generated with a visibility $v$, that is, the state can be expressed as $\rho^{(i)}_{A_iB_i}=v\proj{\phi^+_d}+(1-v)\I/d^2$ for any $i$.

Now, we consider again the functional $\mathcal{S}_{\{E_b\}}$ \eqref{Witunivm}. Using states $\rho^{(i)}_{A_iB_i}$, we obtain as Eq. 
\eqref{eq9} that for any $\eta\in[0,1]$, $\mathcal{S}_{\{E_b\}}$ is
\begin{equation}
   \max_{b}\frac{\Tr(\eta E_b)}{d^N}\left( v^N \Tr(\mathcal{W}_{\{E_b\}}(\I-E_b))-\Tr(\mathcal{W}_{\{E_b\}})\right).
\end{equation}
If the visibility $v=1$, then for any entangled measurement ${\{E_b\}}$, $ \mathcal{S}_{\{E_b\}}>0$ for any $\eta\in(0,1]$. Consequently, when the sources are perfect, swap-steering can be detected using any entangled measurement with arbitrarily low detection efficiency.
In the case of standard quantum steering, the critical detection efficiency increases with the visibility of the state, that is, if the source produces states with lower visibility, more efficient detectors are required to witness quantum steering \cite{Srivastav_2022,Vallone_2013}. Unlike this, we observe that for swap-steering, for any non-zero detection, one can always observe swap-steering when the visibility of each source is above a certain critical visibility $v_c$ given by
\begin{eqnarray}
    v_c=\left(\frac{\Tr(\mathcal{W}_{\{E_b\}})}{\Tr(\mathcal{W}(\I-E_b))}\right)^{1/N}.
\end{eqnarray}

For a realisable entanglement swapping experiment \cite{entexp1,doi:10.1126/sciadv.aea8571}, where $\tilde{E}_0=\proj{\psi-}$, with the sources generating the state $v\proj{\psi_-}+(1-v)\I/4$, and the entanglement witness $\mathcal{W}_{\tilde{E}_0}=\I/2-\proj{\psi-}$, we obtain the the critical visibility to be $v_c=(2/3)^{1/2}\approx 0.81$, thus well within the current experimental limit. Consequently, with realisable practical sources, one can observe quantum steering for any non-zero detection efficiency in a loophole-free way. 

{\it{Discussions---}} In this work, we found that similar to detecting entangled quantum states, entanglement in any quantum measurement can also be detected without trusting them. In particular, we proposed an entanglement witness of the Bell-basis measurement that only requires to produce $\{\proj{0},\proj{1},\proj{+},\proj{-}\}$ which is simple to implement. For instance, in a polarisation-based setup, one has to simply rotate the polarisation of a single state $\ket{0}$ to prepare all the other states. We then extended this idea to witness entanglement in quantum measurements in a 1SDI way. For this purpose, we utilised the swap-steering scenario which is the minimal scenario to witness any form of network nonlocality without inputs. 

The usual approach to reducing the detection efficiency in standard nonlocality or steering tests is either using more number of measurements per site \cite{Miklin_2022,Srivastav_2022, Vallone_2013}, non-maximally entangled state \cite{CH} or full statistics \cite{Masini}. To obtain arbitrarily low critical detection efficiency, one requires infinite number of measurements per site, thus, genuine infinite randomness is required to witness them in a loophole-free way, which is impossible. Moreover, in all these approaches, as the visibility of the states decreases, the critical detection efficiency increases. Counterintuitively, we observed here that for arbitrary non-zero detection efficiency of any single entangled measurement, one can observe swap-steering in the quantum network if the sources generate states beyond a certain critical visibility. A simple way to understand this is to observe that the notion of steering in networks is that the quantum state at the trusted side gets updated or steered from a separable to an entangled state, due to the measurement of Bob. Now, if Bob's measurement possess any finite amount of entanglement, along with the sources generating the maximally entangled state, the entanglement-swapping ensures that the reduced state on the trusted side becomes entangled, which is then detected by Alice.

Several interesting problems follow-up from this work. The first one concerns finding optimal witnesses for several other entangled measurements. A challenging problem in this direction will be to witness more structures in the quantum measurement. For instance, is it possible to detect that every measurement element is entangled or not. The most interesting follow-up problem will be to demonstrate the above results experimentally and utilise the scenario for a cryptographic protocol.

\begin{center}
    \textbf{Acknowledgements}
\end{center}
We thank Emanuele Polino for fruitful discussions. We acknowledge the National Science Centre, Poland, grant Opus 25, UMO-2023/49/B/ST2/02468.

\providecommand{\noopsort}[1]{}\providecommand{\singleletter}[1]{#1}%

\onecolumngrid
\section*{Appendix}
\subsection{SOHS models}
Consider the scenario depicted in Fig. 2 of the manuscript that such that the sources generate the classically correlated state given by $\sum_{k}p_k\bigotimes_{i=1}^N\rho^{(k)}_{A_iB_i}$ and Bob's measurement is $\{M_b\}$ where $M_b=\sum_{j}\bigotimes_{i=1}^N\sigma^{(j)}_{i,b}$ such that 
$b=0,\ldots,n-1$. The assemblage, or collection of post-measured states, $\Xi=\{\xi_b\}$ is given by $\xi_b=\sum_{k,j}p_k\Tr_{B}(\bigotimes_{i=1}^N\I_{A_i}\otimes\sigma^{(j)}_{i,b}\rho^{(k)}_{A_iB_i})$.

Let us now construct an SOHS model that reproduces the same assemblage as $\Xi$. The intuition behind the following construction is that all the sources $S_i$ deterministically send the outcome $b$ to Bob and send the corresponding local states to Alice to reproduce the assemblage $\Xi$. For this purpose, we consider that $\lambda_i=(b_i,j_i)$ where $b_i,j_i$ are non-negative integers such that $b_i=0,\ldots,n-1$. The corresponding correlations $p(a,b)$ under the SOHS model is given as 
\begin{eqnarray}
\sum_{k}p_{k}\sum_{\lambda_1^{(k)},\ldots,\lambda_N^{(k)}}p\left(\lambda_1^{(k)}\ldots,\lambda_N^{(k)}\right)\Tr\left(M_{a|i_1\ldots i_N}\bigotimes_{i=1}^N\tilde{\rho}_{\lambda_i^{(k)}}\right)\nonumber\\p\left(b|\lambda_1^{(k)},\ldots,\lambda_N^{(k)}\right)\qquad
\end{eqnarray}
The probability distribution $p(\lambda_1^{(k)},\ldots,\lambda_N^{(k)})$  for all $k$ is given by
\begin{eqnarray}
    p(\lambda_1^{(k)},\ldots,\lambda_N^{(k)})=\begin{cases}
        p_{b,j,k} & \lambda_1=\ldots=\lambda_N=(b,j)\\
        0  &\mathrm{otherwise}
    \end{cases}
\end{eqnarray}
where $p_{b,j,k}=\Tr(\bigotimes_{i=1}^N\I_{A_i}\otimes\sigma^{(j)}_{i,b}\rho^{(k)}_{A_iB_i})$ along with $p\left(b|\lambda_1^{(k)},\ldots,\lambda_N^{(k)}\right)$ given as
\begin{eqnarray}
p\left(b|\lambda_1^{(k)},\ldots,\lambda_N^{(k)}\right)=\begin{cases}
        1 & \lambda_1=(b,j) \qquad \forall j,k\\
        0  &\mathrm{otherwise}
    \end{cases}.
\end{eqnarray}
The corresponding local states sent by the source $S_i$ to Alice is $\tilde{\rho}_{\lambda_i^{(k)}}=\Tr_B(\I_{A_i}\otimes\sigma^{(j)}_{i,b}\rho^{(k)}_{A_iB_i})/p_{b,j,k}$.

\subsection{Upper-bound for SOHS models}
\setcounter{fakt}{1}
\begin{fakt}
    Consider the swap-steering scenario in Fig. 2 of the manuscript and the functional $\mathcal{S}_{\{E_b\}}$ given by 
    \begin{eqnarray}\label{Wituniv}
    \mathcal{S}_{\{E_b\}}=\max_{b}\sum_{i_1,\ldots i_N}\beta_{i_1,\ldots i_N}p(0,b|i_1\ldots i_N)
\end{eqnarray}
    The maximal value attainable of $\mathcal{S}_{\{E_b\}}$ using an SOHS model is $\beta_{SOHS}=0$. 
\end{fakt}
\begin{proof}
    Let us recall that for correlations admitting a SOHS model, we have that
    \begin{eqnarray}
        p(0,b|i_1\ldots i_N)= \sum_{k}{p_k}\sum_{\lambda_1^{(k)},\ldots,\lambda_N^{(k)}}p(\lambda_1^{(k)})\ldots p(\lambda_N^{(k)})\nonumber\\ \Tr{[M_{0|i_1\ldots i_N}\rho_{\lambda_1^{(k)}}\otimes\ldots\rho_{\lambda_N^{(k)}}]} p(b|\lambda_1^{(k)},\ldots,\lambda_N^{(k)})
    \end{eqnarray}
for all $b$. As $\mathcal{S}_{\{E_b\}}$ is linear over the probabilities $p(0,b|i_1\ldots i_N)$ we can straightforwardly optimise over models with $p_{k'}=1$ for some $k=k'$. Consequently, we have from \eqref{Wituniv} that
\begin{eqnarray}\label{A21}
    \mathcal{S}_{\{E_b\}}=\max_b\sum_{\lambda_1,\ldots,\lambda_N}p(\lambda_1)\ldots p(\lambda_N \Gamma(\rho_{\lambda_1}\otimes\ldots\rho_{\lambda_N}) )\nonumber\\p(b|\lambda_1,\ldots,\lambda_N)\quad
\end{eqnarray}
where 
\begin{eqnarray}\label{gamma11}
    \Gamma(\rho_{\lambda_1}\otimes\rho_{\lambda_n})=\sum_{i_1,\ldots i_N}\beta_{i_1,\ldots i_N}p(0|i_1\ldots i_N,\rho_{\lambda_1}\otimes\ldots\rho_{\lambda_N}).
\end{eqnarray}
Expanding the above formula \eqref{gamma11}, by recalling Alice's measurements $\mathcal{A}_{i_1\ldots i_N}$, we obtain
\begin{eqnarray}
    \Gamma(\rho_{\lambda_1}\otimes\ldots\rho_{\lambda_N})&=& \sum_{i_1,\ldots i_N}\beta_{i_1,\ldots i_N}\Tr(\tau_{i_1}\otimes\ldots\tau_{i_N} \rho_{\lambda_1}\otimes\ldots\rho_{\lambda_N}) \nonumber\\&=&-\Tr(\mathcal{W}\rho_{\lambda_1}\otimes\ldots\rho_{\lambda_N})\leq 0.
\end{eqnarray}
Thus, we have from \eqref{A21} that for correlations admitting a SOHS model
\begin{eqnarray}
    \mathcal{S}_{\{E_b\}}\leq0.
\end{eqnarray}
This concludes the proof.
\end{proof}

\end{document}